\documentclass[a4paper]{article}
\usepackage{color,graphicx}
\usepackage{amssymb,amsmath,amsthm}
\newcommand{\figfile}[1]{#1}%
\theoremstyle{plain}%
\newtheorem{theorem}{Theorem}[]%
\newcommand{\speed}{c_{\mathrm{ph}}}
\newcommand{\grspeed}{c_{\mathrm{gr}}}
\newcommand{\fspace}[1]{{#1}}
\newcommand{\iu}{\mathrm{i}}
\newcommand{\mhexp}[1]{\exp(#1)}

\newcommand{\ccinterval}[2]{\left[#1,\,#2\right]}

\newcommand{\oointerval}[2]{\left(#1,\,#2\right)}

\newcommand{\bigpar}{\par\quad\par}

\newcommand{\phase}{\varphi}

\newcommand{\calR}{\mathcal{R}}
\newcommand{\calS}{\mathcal{S}}
\newcommand{\calV}{\mathcal{V}}

\newcommand{\Zset}{\mathbb{Z}}

\newcommand{\Rset}{\mathbb{R}}
\newcommand{\Cset}{\mathbb{C}}

\newcommand{\dint}{\mathrm{d}}

\newcommand{\sgn}{\mathrm{sgn}\,}
\newcommand{\fin}{\mathrm{fin}}
\newcommand{\osc}{\mathrm{osc}}
\newcommand{\non}{\mathrm{non}}
\newcommand{\comp}{\mathrm{c}}
\newcommand{\per}{\mathrm{per}}

\newcommand{\average}[1]{{\langle{#1}\rangle}}
\newcommand{\baverage}[1]{{\big\langle{#1}\big\rangle}}
\newcommand{\eps}{\varepsilon}

\newcommand{\at}[1]{{({#1})}}
\newcommand{\bat}[1]{{\big({#1}\big)}}
\newcommand{\Bat}[1]{{\Big({#1}\Big)}}
\newcommand{\pair}[2]{{(#1,#2)}}
\newcommand{\bpair}[2]{{\big({#1,#2}\big)}}
\newcommand{\triple}[3]{{\left({#1,#2,#3}\right)}}

\newcommand{\al}{\alpha} 

\newcommand{\ka}{\kappa}
 
\newcommand{\Om}{\Omega}
 
\newcommand{\si}{\sigma}

\newcommand{\del}{\partial}
\newcommand{\R}{{\mathbb R}}
\newcommand{\N}{{\mathbb N}}

\newcommand{\abs}[1]{\left\lvert#1\right\rvert}

\newcommand{\jump}[1]{[\![#1]\!]}
\newcommand{\mean}[1]{\left\{#1\right\}}

\newcommand{\sect}\S
\def\defeq{:=}

\newcommand{\myparagraph}[1]{\paragraph*{#1}}
\newcommand{\url}[1]{\\{\tt #1}}

\usepackage[%
a4paper,%
nohead,%
twoside,
ignorefoot,
scale=0.75,
bindingoffset=.25cm
]{geometry}%
\usepackage{float}%
\begin{document}
\title{On selection criteria for problems with moving inhomogeneities}
\author{%
  Michael Herrmann\footnote{%
  Department of Mathematics, Saarland University, Germany,
  {\tt michael.herrmann@math.uni-sb.de}
  }
\and
Hartmut Schwetlick\footnote{%
  Department of Mathematical Sciences, University of Bath, United Kingdom,
  {\tt h.schwetlick@maths.bath.ac.uk}
  }
  \and
  Johannes Zimmer\footnote{%
  Department of Mathematical Sciences, University of Bath, United Kingdom,
  {\tt zimmer@maths.bath.ac.uk}
  }%
}%
\date{\today}
\maketitle
%
%
\begin{abstract}%

\end{abstract}%
We study mechanical problems with multiple solutions and introduce a
thermodynamic framework to formulate two different selection criteria in 
terms of macroscopic energy productions and fluxes. 
Simple examples for lattice motion are then studied to compare the implications for 
both resting and moving inhomogeneities.
%
%
\quad\newline\noindent%
\begin{minipage}[t]{0.15\textwidth}%
Keywords: %
\end{minipage}%
\begin{minipage}[t]{0.8\textwidth}%
\emph{%
selection criteria, radiation condition, causality principle,  \\
lattice dynamics, Young measures, hyperbolic scaling limits, phase transitions} %
\end{minipage}%
\medskip
\newline\noindent
\begin{minipage}[t]{0.15\textwidth}%
PACS: %
\end{minipage}%
\begin{minipage}[t]{0.8\textwidth}%
31.15.-p, 
31.15.xv, 
61.72.-y, 
62.30.+d  
\end{minipage}%
%
%
%
%
%
%
\section{Introduction}
%
Many equations describing mechanical systems exhibit more than one solution, and a common problem is then to select physically reasonable solutions. The central aim of the paper is to analyse selection conditions for problems with 
inhomogeneities, such as sources or phase interfaces.
\par%
In the literature one finds three commonly cited selection criteria. Two of them are reminiscent of Sommerfeld, who gives two physical reasonings of his \emph{radiation condition} for standing sources  \cite{Sommerfeld:49a,Sommerfeld:62a}. He first  states that sources have to be sources, and then stipulates that energy cannot flow in from infinity. The third selection criterion is the \emph{causality principle} introduced by Slepyan \cite[Appendix A]{Slepyan:01a}, see also the monograph~\cite{Slepyan:02a}, which determines the admissible group speeds on both sides of the inhomogeneity.  All these criteria have been applied mainly in the context of Fourier analysis. 
\par
In this paper, we present a different approach to selection criteria,
based on thermodynamic fields, which are derived
via  hyperbolic scaling limits and Young measures and describe the underlying 
mechanical systems on large spatial and temporal scales. In particular, we show that
Sommerfeld's two formulations of his radiation condition are naturally related to rather different fields such as \emph{production of oscillatory energy} and \emph{radiation flux}.  Our framework covers arbitrary nonlinear Hamiltonian lattices and PDEs, but we have to specialise to linear or almost linear cases in order
to obtain explicit expressions for the thermodynamic fields. 
\par%
While we recover Sommerfeld's results for standing sources,
the case of a moving phase interface turns out to be remarkably different. Specifically, in the first case we find that both thermodynamic selection criteria are equivalent and agree with the causality principle. In the second case, however,
the condition on radiation (energy) fluxes can be in contradiction to the entropy inequality, the production criterion and the causality principle. Thus
the energy flux criterion should not be applied uncritically to problems with moving inhomogeneities.
\par
We present our exposition along two guiding examples of nearest neighbour (NN) chains of atoms. The governing equations are
\begin{align}
  \label{Eqn:ForcedNN}
  \ddot{x}_j\at{t}=
  \Phi^\prime\at{x_{j+1}\at{t}-x_j\at{t}}-\Phi^\prime\at{x_{j}
    \at{t}-x_{j-1}\at{t}} +\zeta\at{t}\delta_{j0},
\end{align}
where $\zeta$ is the external forcing of atom $j=0$, and the two examples are $\at{i}$ a harmonic chain with periodic forcing
$\zeta$, $\at{ii}$ a moving phase interface in a chain with bi-quadratic interaction potential $\Phi$ and no external forcing.

\myparagraph{Sommerfeld's radiation condition}

Sommerfeld's approach to radiation conditions is described in his
book~\cite[\S28]{Sommerfeld:62a}, see also~\cite{Sommerfeld:49a}, and we thus just describe
the gist of some key arguments. However, the discussion given in
\sect\ref{sec:Stat-Somm-probl} for the case of a lattice model closely resembles the original spatially
continuous case, and we refer the reader to that section for a more detailed
mathematical treatment.

Sommerfeld considers the wave equation on $\R^3$ with a temporally periodic forcing at the
origin. A separation of variables ansatz leads to the Helmholtz equation with a source at the
origin, thus an inhomogeneous equation. The corresponding homogeneous equation has nontrivial
bounded kernel functions (in $\R^3$). Some kernel functions are excluded by boundary or symmetry
conditions, in Sommerfeld's case by requiring that the solution is radially symmetric and decaying at infinity. This still does not single out a
unique solution.  In particular, there are two singular and decaying solutions $u^\pm$ to the
inhomogeneous Helmholtz equation, which are well-defined outside the origin. Let $x^\pm$ be
the corresponding solutions of the forced wave equation. The choice of the direction of time
made in the separation of variables \emph{ansatz} renders the corresponding solution $x^+$
an \emph{outwardly} radiating radial wave and $x^-$ an \emph{inwardly} radiating wave. In
fact, $x^+$ describes a \emph{source} solution to the forced wave equation in the sense that the
forcing supplies energy to the system. On the contrary, $x^-$ corresponds to a sink solution,
as the forcing now deprives the system of energy.

Sommerfeld now introduces a \emph{binary} choice, allowing only waves which propagate
outwards and dismissing those which propagate inwards. This selection is necessary both for
physical and mathematical reasons: mathematically, the choice is an integral part of the
arguments leading to a unique fundamental solution of the forced wave equation. Physically the two solutions $x^\pm$ are qualitatively very different.  In Sommerfeld's
words~\cite{Sommerfeld:62a},
\begin{equation}
  \label{eq:SOM1}
  \tag{SOM1}
  \begin{minipage}{12cm}
    ``Quellen sollen \emph{Quellen}, nicht Senken der Energie sein.''     [Sources have to be
\emph{sources}, not sinks of the energy.]
  \end{minipage}
\end{equation}
We call this the \emph{first formulation} of Sommerfeld's radiation condition. He then gives
what we call the \emph{second formulation},
\begin{equation}
  \label{eq:SOM2}
  \tag{SOM2}
  \begin{minipage}{12cm}
    ``Die von den Quellen ausgestrahlte Energie mu\ss\ sich ins     Unendliche zerstreuen,
\emph{Energie darf nicht aus dem       Unendlichen in die vorge\-schriebenen Singularit\"aten
des       Feldes eingestrahlt werden}.'' [The energy radiated from the     sources has to scatter
to infinity, \emph{energy must not be       radiating from infinity into the prescribed
singularities of the       field}.]
  \end{minipage}
\end{equation}

\myparagraph{Thermodynamic approach }

In \sect\ref{sec:Macr-field-equat}  we establish a thermodynamic framework to analyse~\eqref{eq:SOM1} and~\eqref{eq:SOM2}. The starting point are \emph{discrete local conservation laws}, which can be derived from 
\eqref{Eqn:ForcedNN} and describe the transport of mass, momentum and energy on the \emph{microscopic scale}. In a first step we introduce a small scaling parameter $\eps$ and use the concept of \emph{Young measures} to establish the macroscopic limit $\eps\to0$. In this limit the microscopic conservation laws converge to their macroscopic counterparts, which are PDEs with densities, fluxes and productions. 
\par
In a second step we identify the macroscopic quantities
that admit a thermodynamic interpretation of Sommerfeld's radiation condition. 
The \emph{radiation flux} $Q$ is defined as the Galilean invariant part of the energy flux and~\eqref{eq:SOM1} is naturally related to the sign of $Q$ on both sides of the inhomogeneity. We also introduce 
the \emph{oscillatory energy} $E_\osc$, which measures 
the amount of macroscopic energy stored in microscopic oscillations, and
the \emph{non-oscillatory energy} $E_\non$. Sommerfeld's first formulation of the radiation condition can now be reinterpreted in terms of energy productions. More precisely, for forced excitation problems as studied by Sommerfeld for the wave equation, \eqref{eq:SOM1} is naturally related to the sign of the local production of total energy $E=E_\osc+E_\non$ due to the external forcing.
The situation is different for phase transition waves as these are not driven by external forces and conserve the total energy exactly (in a local sense). Instead, the key energetic phenomenon is now that the propagating phase interface causes a constant transfer between oscillatory and non-oscillatory energy. This transfer gives rise to a production of oscillatory energy, and~\eqref{eq:SOM1} is naturally related to the sign of this production.
\par
Notice that the above arguments are solely based on energy productions and fluxes, that means on fundamental concepts from continuum mechanics and thermodynamics. They can therefore, at least in principle,
be also applied to more general systems, as for instance NN chains with generic double-well potentials. However, only in some special cases it is also possible to derive \emph{explicit} expressions for all densities, fluxes and production terms. 
\par
We finally emphasise that our macroscopic interpretations of \eqref{eq:SOM1} and \eqref{eq:SOM2} provide only \emph{necessary} conditions. To ensure
uniqueness, in particular to exclude linear combinations of $x^+$ and $x^-$ in Sommerfeld's forced excitation problem, one typically imposes \emph{microscopic selection criteria}, which also exploit non-macroscopic properties of solutions. Examples are the asymptotic radiation condition~\eqref{eq:som-3-lat}  or, equivalently, the causality principle. Both methods are inextricably connected with Fourier methods. It is therefore not clear how to generalise these selection criteria to nonlinear problems. 
\par
%

\myparagraph{Results for the special examples}
\par

In \sect\ref{sec:Stat-Somm-probl} we apply the thermodynamic framework to the analogue of
Sommerfeld's problem in the harmonic NN chain and find that \eqref{eq:SOM1} and \eqref{eq:SOM2} are
equivalent and comply with the microscopic selection criterion. Afterwards we study phase transition waves in \sect\ref{sec:Radi-cond-trav} and
show that both conditions provide different selection criteria for moving phase interfaces.
We show that~\eqref{eq:SOM1} is equivalent to the usual entropy inequality, and thus
we propose to apply~\eqref{eq:SOM1} to problems with moving inhomogeneities as well.  The application of~\eqref{eq:SOM2}, however, turns out to be problematic. Specifically, we show that for subsonic wave speed close to the speed of sound, the selection criterion~\eqref{eq:SOM2} selects a phase transition wave that \emph{violates} the entropy inequality as well as the causality principle. We thus conclude that~\eqref{eq:SOM2} has to be discarded, at least for some problems with moving inhomogeneities. 
%
%
\section{Macroscopic field equations in the presence of microscopic oscillations}
\label{sec:Macr-field-equat}

Our approach to Sommerfeld's radiation condition in \sect\ref{sec:Stat-Somm-probl} and
\sect\ref{sec:Radi-cond-trav} is based on macroscopic balance laws that
govern the effective dynamics of averaged quantities on large spatial and temporal scales.
In this section we derive and discuss these balance laws and provide the atomistic
expressions for all densities, fluxes and production terms. Our exposition is formal but we
emphasise that all arguments can be made rigorous using Young measures; see
\cite{DH08,AMSMSP:DHR} and Appendix~\ref{app:1}. For the sake of clarity, we start with the
conservation laws for the unforced NN chain, that is~\eqref{Eqn:ForcedNN} with $\zeta=0$.

\myparagraph{Microscopic conservation laws}

We can rewrite~\eqref{Eqn:ForcedNN} with $\zeta=0$ in terms of the atomic distances
(discrete strains) $r_j\defeq x_{j+1}-x_j$, and velocities $v_j\defeq \dot x_j$ as first order
equations
\begin{align}
  \label{eq:DCL.MassMom}
  \dot{r}_j=v_{j+1}-v_{j},\quad
  \dot{v}_j=\Phi^\prime\at{r_j}-\Phi^\prime\at{r_{j-1}},
\end{align}
which can be viewed as the discrete counterparts of the local conservation laws for mass and
momentum in Lagrangian coordinates, see~\eqref{Eqn:Macro.Cons.Laws}.
We note that~\eqref{eq:DCL.MassMom} constitute \emph{local conservation laws} in the sense that the time derivate of some atomistic observable equals the discrete divergence of another observable. As a consequence we find, assuming appropriate decay  conditions for $j\to\pm\infty$, that each local conservation law implies a global one by summing over $j$. The discrete divergence structure is also an important ingredient for the thermodynamic limit because it allows us to derive the macroscopic conservation laws, see Appendix~\ref{app:1}.
\par%
Since the unforced NN chain is
an autonomous Hamiltonian system with shift symmetry, we can also derive a local conservation law for the energy, namely
\begin{align}
  \label{eq:DCL.Energy}
  \tfrac{\dint}{\dint{t}}\bat{\tfrac{1}{2}v_j^2+\Phi\at{r_{j-1}}}
  =
  v_j\Phi^\prime\at{r_j}-v_{j-1}\Phi^\prime\at
  {r_{j-1}}.
\end{align}
To characterise the thermodynamic properties of NN chains, we now derive a macroscopic
description by applying the \emph{hyperbolic scaling} of space and time. For a given scaling
parameter $0<\eps\ll{1}$, we define the
\emph{macroscopic time} $\tau$ and the \emph{macroscopic particle index} $\xi$ by
\begin{align}
  \label{eq:NN.Scaling}
  \tau=\eps{t},\quad\xi=\eps{j},
\end{align}
but we do \emph{not} scale distances and velocities. We then regard the atomistic data that
correspond to a solution of~\eqref{eq:DCL.MassMom} as functions $r_\eps$ and $v_\eps$ that depend continuously on
$\tau$ and are piecewise constant in $\xi$. That is, we identify
$r_j\at{t}=r_\eps\pair{\eps{t}}{\eps{j}}$ and $v_j\at{t}=v_\eps\pair{\eps{t}}{\eps{j}}$.

Of course, the functions $r_\eps$ and $v_\eps$ will, in general, be highly oscillatory with
wave length of order $\eps$ and do not converge as $\eps\to0$ in a pointwise sense. However,
as long as the solution to~\eqref{eq:DCL.MassMom} is \emph{bounded} we can assume, thanks to
weak compactness, that for any \emph{atomic observable} $\psi=\psi\pair{r}{v}$ the functions
$\psi\pair{r_\eps}{v_\eps}$ converge weakly as $\eps\to0$,  see Appendix~\ref{app:1}. The
limit function $\average{\psi}$ is then non-oscillatory and can be regarded as the
\emph{thermodynamic  field} of $\psi$, which means that $\average{\psi}\pair{\tau}{\xi}$ gives the
local mean value of $\psi$ in the macroscopic point $\pair{\tau}{\xi}$.

\myparagraph{Macroscopic conservation laws}

In the thermodynamic limit $\eps\to0$ the discrete conservation laws~\eqref{eq:DCL.MassMom}
and~\eqref{eq:DCL.Energy} transform into
\begin{align}
  \label{eq:NN.Limit}
  \partial_\tau\average{r}=\partial_\xi\average{v},\quad
  \partial_\tau\average{v}=\partial_\xi\baverage{\Phi^\prime\at{r}},\quad
  \partial_\tau\baverage{\tfrac{1}{2}v^2+\Phi\at{r}}
  =\partial_\xi\baverage{v\Phi^\prime\at{r}}.
\end{align}
These PDEs describe the local conservation laws for mass, momentum and energy on the
macroscopic scale and are well known within the thermodynamic theory of elastic bodies. In
fact, they can be written as
\begin{align}
  \label{Eqn:Macro.Cons.Laws}
  \partial_\tau{R}-\partial_\xi{V}=0,\quad
  \partial_\tau{V}+\partial_\xi{P}=0,\quad
  \partial_\tau{E}+\partial_\xi{F}=0,\quad
\end{align}
with \emph{macroscopic strain} $R=\average{r}$, \emph{macroscopic   velocity}
$V=\average{v}$, \emph{pressure} $P=-\baverage{\Phi^\prime\at{r}}$, \emph{total energy density}
$E=\baverage{\tfrac{1}{2}v^2+\Phi\at{r}}$, and \emph{total energy flux}
$F=-\baverage{v\Phi^\prime\at{r}}$. Moreover, splitting off the Galilean invariant part from
both the energy density and the energy flux, we find
\begin{align}
  \label{eq:split-Gali}
  E=\tfrac{1}{2}V^2+U,\quad
  F=VP+Q,
\end{align}
with \emph{internal energy density} $U$ and heat flux $Q$. \emph{Radiation}
in the sense of Sommerfeld precisely means energy transport via $Q$, see also
\sect\ref{sec:Stat-Somm-probl}, and therefore we call $Q$ the \emph{radiation flux}.

We emphasise that, in general, the conservation laws~\eqref{Eqn:Macro.Cons.Laws} do
\emph{not} constitute a closed system, but must be accompanied by \emph{closure relations}.
Unfortunately, very little is known about the thermodynamic limit for most initial data and
general interaction potentials. In some cases, however, it is possible to solve the closure
problem. Below we show that all thermodynamic fields can be computed explicitly for $\at{i}$
Sommerfeld's fundamental solution in forced harmonic NN chains and $\at{ii}$ phase transition
waves in NN chains with bi-quadratic potential.

\myparagraph{Oscillatory energy}

For our purposes, the split~\eqref{eq:split-Gali} is not sufficient; it is convenient to introduce another split of the total energy density $E$ that
accounts for the fact that the computations of local mean values and nonlinearities do not
commute in the presence of oscillations. To obtain a precise measure for the strength of the
oscillations we write
\begin{align*}
  E=E_\non+E_\osc
  ,\qquad
  E_\non=\tfrac{1}{2}V^2+\Phi\at{R},
\end{align*}
which means
\begin{align}
  \label{eq:PartialEnergies}
  E_\non=\tfrac{1}{2}\average{v}^2+\Phi\bat{\average{r}}
  ,\qquad%
  E_\osc=\tfrac{1}{2}\baverage{\bat{v-\average{v}}^2}
  +\baverage{\Phi\at{r}-\Phi\bat{\average{r}}}
\end{align}
and furthermore implies $U=\Phi\bat{\average{r}}+E_\osc$. We refer to $E_\osc$ and $E_\non$ as oscillatory and non-oscillatory energy density,
respectively, and emphasise that $E_\osc$ measures precisely the amount of macroscopic energy
that is locally stored within the oscillations.  From \eqref{eq:NN.Limit} and
\eqref{eq:PartialEnergies} we now conclude that the partial energies are balanced by
\begin{align}
  \label{eq:production.e1}
  \partial_\tau{E_\osc}+\partial_{\xi}Q=\Xi,\qquad
\partial_\tau{E_\non}+\partial_{\xi}\at{PV}=-\Xi,
\end{align}
where the production terms are given by 
\begin{align*}
  {\Xi}=-\at{P+\Phi^\prime\at{R}}\partial_\xi{V}
  =\Bat{\baverage{
      \Phi^\prime\at{r}}-\Phi^\prime\bat{\average{r}}}\partial_\xi\average{v}\,.
\end{align*}
The physical meaning of $\Xi$ becomes apparent when we compute the
temporal change of the oscillatory and non-oscillatory energy contained in a control volume $\Om=\ccinterval{a}{b}$. This
 gives
\begin{align*}
\frac{\dint}{\dint{\tau}}\int_a^b E_\osc\dint\xi+Q|^{\xi=b}_{\xi=a}=
\int_a^b \Xi\,\dint\xi=-\frac{\dint}{\dint{\tau}}\int_a^b E_\non\dint\xi-\at{PV}|^{\xi=b}_{\xi=a}\,.
\end{align*}
Since the flux terms $Q$ and $PV$ are related to the exchange of energy with the exterior of $\Om$, we conclude that $\Xi\pair{t}{\xi}$ measures  precisely how much non-oscillatory energy is transferred into oscillatory energy at time $\tau$ at the point $\xi$. 
\par%
We further mention that the second law of thermodynamics is, at least on an intuitive level, related to 
a sign condition for $\Xi$. While each scaling limit of the unforced NN chain conserves the total energy
according to \eqref{eq:NN.Limit}$_3$, there exist many solutions that yield a non-vanishing field $\Xi$ in the limit $\eps\to0$.
For $\Xi\geq0$ these solutions `dissipate' non-oscillatory energy (at the expense of increasing oscilllatory energy);
this is in perfect agreement with our physical intuition since
oscillations can be regarded as microscopic fluctuations. If time is
reversed, however, $\Xi$ becomes negative and the solution extracts
energy from the fluctuations, which contradicts the intuition. This is
discussed in more detail in~\sect\ref{sec:Stat-Somm-probl} and~\sect\ref{sec:Radi-cond-trav}.

\myparagraph{Thermodynamic fields for harmonic oscillations}

 With the central quantities $E_\osc$ and $E_\non$ now being defined, we can return to the analyis of thermodynamic fields. Motivated by the discussion in \sect\ref{sec:Stat-Somm-probl} and
\sect\ref{sec:Radi-cond-trav}, we restrict the analysis and compute the thermodynamic fields for \emph{travelling
waves} in harmonic NN chains.  A travelling wave for the NN chain is an exact solution
to~\eqref{Eqn:ForcedNN} with $\zeta=0$ that satisfies
\begin{align}
  \label{TWEqn}
  r_j\at{t}=\calR\at{j-\speed{t}},\quad
  v_j\at{t}=\calV\at{j-\speed{t}}
\end{align}
for some \emph{phase speed} $\speed$ and \emph{profile functions} $\calR$ and $\calV$ that
depend on the phase variable $\phase=j-\speed{t}$. Travelling waves in NN chains are determined by
advance-delay differential equations, see~\cite{Filip:99a,Dreyer:04a,Her10a}, and describe
fundamental \emph{oscillatory patterns}. 
\par
The hyperbolic scaling \eqref{eq:NN.Scaling} implies that each travelling wave solution for the harmonic NN chain \eqref{TWEqn} converges as $\eps\to0$ to a Young measure that is constant in space and time. The thermodynamic fields of all observables $\psi$ are therefore independent of $\pair{\tau}{\xi}$ and can be computed from $R$ and $V$ as follows 
\begin{align}
  \label{Eqn:TW.MeanValues}
  \average{\psi}\defeq\lim\limits_{L\to\infty}
  \frac{1}{2L}\int_{-L}^{L}
  \psi\pair{\calR\at\phase}{\calV\at\phase}
  \,\dint\phase.
\end{align}
\par
For a harmonic NN chain, which has the  interaction potential $\Phi(r) =
\tfrac{1}{2}{c_0^2}{r^2}+d_1{r}+d_0$, we immediately verify by Fourier transform that for
prescribed $\speed$ with $0<\abs{\speed}<c_0$ travelling waves are given by
\begin{align}
  \begin{split}
    \label{Eqn:TW.Harmonic}
    \calR\at{\phase}=R
    +\sum_{i=1}^M{A}_i\cos\at{\kappa_i\phase+{\kappa_i/2}+\eta_i},\qquad
    \calV\at{\phase}=V\mp{c_0}
    \sum_{i=1}^M{A}_i\cos\at{\kappa_i\phase+\eta_i},
  \end{split}
\end{align}
with ``$-$'' and ``$+$'' for left and right moving waves, respectively, that is, for
$\speed<0$ and $\speed>0$, respectively. Here the wave numbers $\kappa_i$, $i=1, \dots, M$, denote
the positive solutions to $\speed^2{k}^2=\Omega\at{k}^2$, where
\begin{math}
  \Omega\at{k}= 2 c_0 \sin\at{k/2}
\end{math} %
is the dispersion relation of the harmonic NN chain. In particular, near sonic waves with
$\speed\approx{c_0}$ have $M=1$ and depend, up to the phase shift $\eta_1$, on the four independent
parameters $R$, $V$, $A=A_1$, and $\kappa=\kappa_1$.
\par
The thermodynamic fields for such harmonic travelling waves can easily be computed by
\eqref{Eqn:TW.MeanValues} and \eqref{Eqn:TW.Harmonic}. In fact, thanks to $\average{r}=R$ and
$\average{v}=V$, we find
\begin{align}
  \label{eq:P-harm}
  P &= 
  -c_0^2R-d_1
  ,\quad%
  E_\non=
  \tfrac12{V}^2 +\tfrac{1}2c_0^2{R}^2+d_1{R}+d_0
  ,\quad%
\end{align}
as well as
\begin{align}
  \label{TDTW:OscDensAndFlux}
  E_\osc=
  \tfrac{1}{2}c_0^2A^2,\quad
  Q&=\pm\Bat{c_0\sum_{i=1}^M\frac{A_i}{A}\cos\at{\kappa_i/2}}E_\osc
\end{align}
with $A^2=\sum_{i=1}^M{A}_i^2$. For periodic waves with $M=1$ we therefore have
$Q=\grspeed{E_\osc}$, where $\grspeed=\pm\abs{\Omega^{\prime}\at{\kappa}}$ is the \emph{group
speed}.
\par
We emphasise that the thermodynamic computations presented above can be extended to
superpositions of finitely many harmonic travelling waves, with obvious modifications.  We also mention
that a complete characterisation of the energy transport in harmonic lattices can be derived
in terms of Wigner-Husimi measures~\cite{Mielke:06a}.
\myparagraph{Macroscopic description of forcing}

In order to generalise the formalism from above to forced NN chains we assume, for
simplicity, that the forcing acts only in the particle $j=0$, see~\eqref{Eqn:ForcedNN}, and
that $\zeta$ is periodic with
\begin{align}
  \label{App:Forcing}
  \zeta\at{t}=\zeta\at{t+t_\per},\qquad
  \int_0^{t_\per}\zeta\at{t}\,\dint{t}=0.
\end{align}
These conditions guarantee that the forcing does not contribute to the macroscopic
conservation laws for mass,~\eqref{eq:DCL.MassMom}$_1$, and
momentum,~\eqref{eq:DCL.MassMom}$_2$. The forcing, however, in general supplies some energy to
the system, and hence the conservation law~\eqref{eq:DCL.MassMom}$_3$ must be replaced by
\begin{align}
  \label{Eqn.EnergyBalanceWithForcing}
  \partial_\tau{E}+\partial_\xi{F}=
  \theta\at{\tau}\delta_0\at{\dint\xi}.
\end{align}
The \emph{macroscopic energy production} at $\xi=0$ can be computed either as the jump of the
macroscopic energy flux at $\xi=0$ or by averaging the microscopic energy production. This reads
\begin{align}
  \label{eq:production2}
  \theta\at{\tau}=F\pair{\tau}{0+}-F\pair{\tau}{0-}=
  \lim_{\delta\to0}\lim_{\eps\to0}\frac{\eps}{2\delta}
  \int_{\at{\tau-\delta}/\eps}^{\at{\tau+\delta}/\eps}
  v_0\at{t}\zeta\at{t}\,\dint{t}.
\end{align}

%
\section{The forced harmonic NN chain}
\label{sec:Stat-Somm-probl}

Here we present the analogue to Sommerfeld's classical problem in harmonic NN chains, that
is, the localised forced excitation problem
\begin{align}
  \label{eq:som-lat-eq}
  \at{\del_t^2-c_0^2\Delta_1} x_j(t)=\zeta(t)\delta^0_j,
\end{align}
where $\Delta_1$ denotes the discrete Laplacian $\Delta_1 x_j\defeq x_{j+1}+x_{j-1} -2x_j$.
For simplicity we normalise the speed of sound to $c_0=1$ and assume that the chain is
periodically forced at one of its eigenfrequencies $\si$ with $0<\si<2$.  The analysis presented in this section resembles that of the spatially continuous case studied by Sommerfeld~\cite{Sommerfeld:62a}; the purpose of this section is to familiarise readers with our approach to selection criteria, and to show that we recover Sommerfeld's classical results for a standing source in a spatially discrete medium.

\myparagraph{Explicit solutions via Helmholtz equation}

The separation of variables \emph{ansatz} $x_j\at{t}=\mathrm{Re}\at{u_j\mhexp{-\iu\si{t}}}$
transforms \eqref{eq:som-lat-eq} into the discrete Helmholtz equation
\begin{align}
  \label{Eqn:NN.Inhom.Helmholtz}
  \si^2u_j+{\Delta_1}u_j=\delta_{0},
\end{align}
which can be solved by Fourier transform. There exist two \emph{special solutions} $u^+$ and
$u^-$ defined by
\begin{align}
  \notag
   u_j^\pm=\pm\frac{\exp\at{\pm\iu\kappa\abs{j}}}%
   {\iu2\Om\at{\kappa}\Om^\prime\at{\kappa}},\qquad
\end{align}
where $\kappa=\kappa\at{\sigma}$ denotes the unique solution to
\begin{align}
  \notag
  \si^2=\Om\at{k}^2, \qquad \Om\at{k}=2\sin\at{k/2}, \qquad
  0<k<\pi.
\end{align}
\par
Of course, the special solutions $u^-$ and $u^+$ can be linearly combined and also
superimposed by plane waves with wave numbers $\pm\kappa$, which are the kernel functions of
the discrete Helmholtz operator. The general solution to \eqref{Eqn:NN.Inhom.Helmholtz} can
therefore be parameterised by $\alpha,\beta\in\Cset$ as
\begin{align}
\label{eq:som-lat.GeneralSolution}
u_j = u_j\pair{\al}{\beta} = u^+_j + \alpha\exp\at{-i\kappa j} + \beta \exp\at{+ i\kappa j}.
\end{align}
Note that in particular $u^+_j = u_j(0,0)$ and $u^-_j =u_j(\al^-, \beta^-)$ with
$\al^-=\beta^-=-\bat{\iu2\Om\at\ka\Om'\at\ka}^{-1}$.
\par
Sommerfeld's approach to the radiation condition can be viewed as the endeavour to remove the
non-uniqueness and to single out a unique choice for $\alpha$ and $\beta$. To this end, he
introduces a microscopic selection criterion, whose analogue in the harmonic NN chains reads
\begin{align}
  \label{eq:som-3-lat}
  \lim_{j\to+\infty}
\Bat{\frac{\dint{u}_{j}}{\dint j}-\iu\ka{u_{{j}}}}=0, \qquad
\lim_{j\to-\infty}\Bat{\frac{\dint{u}_{j}}{\dint j}+\iu\ka{u_{{j}}}}=0,
\end{align}
and implies that $\alpha=\beta=0$ in \eqref{eq:som-lat.GeneralSolution}. In
particular, the microscopic
radiation condition selects $u^+$ but rules out $u^-$.

\myparagraph{Macroscopic aspects of Sommerfeld's radiation condition}

We now show that the binary choice between $u^-$ and $u^+$ can be understood in terms of
purely macroscopic conditions on the production of oscillatory energy and the direction of
the radiation fluxes. As remarked in \sect\ref{sec:Macr-field-equat}, the thermodynamic
framework can be extended to superpositions of harmonic waves. It is thus 
possible to characterise all bounded solutions to \eqref{eq:som-lat-eq}, in particular kernel
functions and their superpositions. The result of such an analysis is that the thermodynamic
interpretation of \eqref{eq:SOM1} and \eqref{eq:SOM2} rejects superpositions of waves as long
as their influx contribution exceeds the outward contribution. Thus, a half space of all
bounded solutions is rejected, and a half-space accepted. We show this analysis here in
detail for the two extreme cases corresponding to $u_-$ and $u_+$. The analysis for the other
solutions is, \emph{mutatis mutandis}, analogous yet more complicated terms arise.

\par%
At first we notice that $u^+$ and $u^-$ correspond to the real-valued displacements
\begin{align}
  \notag
  x_j^\pm(t) = 
  \frac{\sin(\kappa\abs{j} \mp \sigma t)}{2\Om(\kappa) \Om'(\kappa)}.
\end{align}
Each of these solutions to \eqref{eq:som-lat-eq} consists of two counter-propagating
travelling waves that are glued together at $j=0$, where the travelling waves propagate
towards and away from the inhomogeneity at $j=0$ for $x^+$ and $x^{-}$, respectively. We also
notice that $x^+$ and $x^-$ transform into each other under time reversal, and that they define the atomic distances and velocities
\begin{align}
\label{eq:som-lat-sol.DistVel}
v^\pm_j\at{t}=\mp{A}\cos\at{\kappa\abs{j}\mp\sigma{t}},\qquad
r^\pm_j\at{t}=A
\cos\at{\kappa\abs{j+\tfrac12}\mp\sigma{t}}
\end{align}
where the amplitude $A>0$ is given by $1/A=2\Om^\prime\at{\kappa}=2\cos\at{\kappa/2}$.
\begin{figure}[t!]
\centering{%
\includegraphics[width=0.45\textwidth]%
{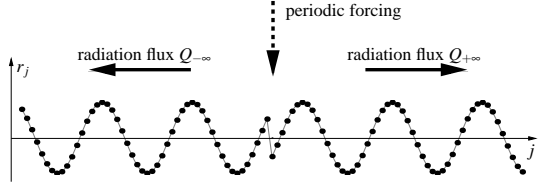}}%
\caption{%
Sommerfeld's \emph{source solution} for the harmonic NN chain: The energy pumped in by the periodic
forcing is radiated towards both $+\infty$ and $-\infty$. If time is reversed
the source becomes a sink and the radiation fluxes on both sides change their sign.}%
\label{fig:cartoon_1}
\end{figure}
\par
To determine the thermodynamic limit $\eps\to0$, the key observation is that both $x^-$ and $x_-$
generate Young measures that are $\at{i}$ independent of the macroscopic time
$\tau=\eps{t}$, $\at{ii}$ constant for $\xi<0$ and $\xi>0$, where $\xi = \eps{j}$ is the macroscopic particle index, and $\at{iii}$ generated by periodic travelling waves. These assertions follow directly from \eqref{eq:som-lat-sol.DistVel} and the definition of Young measure convergence, see \sect\ref{sec:Macr-field-equat} and Appendix \ref{app:1}. They also
imply that the macroscopic conservation laws \eqref{Eqn:Macro.Cons.Laws} are trivially satisfied for $\xi\neq0$. The relevant macroscopic information is therefore encoded in jump conditions for the thermodynamic fields at $\xi=0$. 

Using \eqref{eq:P-harm} and \eqref{TDTW:OscDensAndFlux} we now conclude that almost all
thermodynamic fields are globally constant with
\begin{align*}
  R={0}
  ,\quad%
  V=0
  ,\quad%
  P={0}
  ,\quad%
  F={Q}
  ,\quad%
  E_\non={0}
  ,\quad%
  E={E_\osc}
  ,\quad%
  E_\osc=\tfrac{1}{2}A^2,
\end{align*}
while the radiation flux $Q$ is piecewise constant,
$Q\pair{\tau}{\xi}=\pm\mathrm{sign}\at{\xi}\Om^\prime\at\kappa{E_\osc}=:Q_{\sgn\at{\xi}\infty}$.
The two values for $Q$ are given by
\begin{align*}
  Q_{-\infty}=\mp\tfrac{1}{4}A
  ,\quad%
  Q_{+\infty}=\pm\tfrac{1}{4}A,
\end{align*}
and computing the macroscopic energy production by averaging, see \eqref{eq:production2}, we
find
\begin{align*}
  \theta=
  -\frac{\si}{2\pi}\int\limits_0^{2\pi/\si}v_0\at{t}\cos\at{\si{t}}\,\dint{t}=
  \pm\frac{\si{A}}{2\pi}\int\limits_0^{2\pi/\si}\cos\at{\sigma{t}}^2\,\dint{t}=
  \pm\tfrac{1}{2}A.
\end{align*}
Sommerfeld's first condition~\eqref{eq:SOM1} is naturally related to the sign of $\theta$.
For $x^+$ we have $\theta>0$, so the forcing pumps in energy at $j=0$ and the solution
describes a \emph{source} of the energy. The solution $x^-$, however, corresponds to a
\emph{sink} as $\theta<0$ implies that energy flows out constantly at $j=0$. Moreover, for
the solutions at hand the balance of total energy reduces to
\begin{align*}
  \theta=Q_{+\infty}-Q_{-\infty}=2Q_{+\infty}=-2Q_{-\infty},
\end{align*}
and implies that Sommerfeld's first and second formulation of the radiation condition are
equivalent. Namely, energy that is pumped in at $j=0$ must be
radiated away, and hence the radiation fluxes must point away from $j=0$ , see
Fig.~\ref{fig:cartoon_1}.  Conversely, energy that is deprived from the system at $\xi=0$
must be radiated in from $\pm\infty$. We thus conclude that both~\eqref{eq:SOM1} and~\eqref{eq:SOM2} select the solution $x^+$ but reject $x^-$.

We close this section by mentioning that the \emph{causality principle}, see~\cite{Slepyan:01a,Slepyan:02a} and the discussion in \sect\ref{sec:Radi-cond-trav}, requires
that all oscillatory modes in front and behind the interface satisfy $\grspeed > \speed$ and  $\grspeed < \speed$, respectively. It is not hard to see that this principle also favours $x^+$ and, more generally, implies~\eqref{eq:SOM2} for a standing source.
%
%
\section{Phase transition waves for a bi-quadratic NN chain}
\label{sec:Radi-cond-trav}

We now consider travelling waves with a moving inhomogeneity and apply the selection
criteria~\eqref{eq:SOM1} and~\eqref{eq:SOM2} to these waves. We mention, however, two important caveats of our
analysis: firstly, we assume the existence of subsonic travelling waves with a single
inhomogeneity (interface). Guidance for the existence can be taken
from~\cite{Truskinovsky:05a}; yet existence is a subtle issue, and a rigorous existence proof is available only for a small regime of
subsonic velocities~\cite{Schwetlick:07a,Schwetlick:08a};  it is also worth mentioning that there is a velocity regime where no travelling waves
with a single interface can exist~\cite{Schwetlick:10a}. Secondly, the selection criteria
that result from the thermodynamic interpretation of~\eqref{eq:SOM1} and~\eqref{eq:SOM2} are \emph{necessary} but not sufficient.
\par
\begin{figure}[t!]
  \centering{%
    \includegraphics[width=0.35\textwidth]%
    {\figfile{front_diagramm}}}%
  \caption{Phase transition waves with $\abs{\speed}>c_2$ have periodic     tails and come in
two different types. Type-I waves have either     $0<\grspeed<\speed$ or $0>\grspeed>\speed$,
whereas type-II waves     correspond to either $\grspeed<0<\speed$ or $\grspeed>0>\speed$.   }%
  \label{fig:types}
\end{figure}
We now study heteroclinic solutions to the travelling wave equation~\eqref{TWEqn}. To calculate
the thermodynamic fluxes explicitly, we restrict our considerations to the NN chain with piecewise
quadratic interaction potential
\begin{align}
  \label{TDMF:Potential}
  \Phi\at{r}=\tfrac{1}{2}\min\left\{\at{r-1}^2,\,\at{r+1}^2\right\},
\end{align}
but mention that our thermodynamic arguments can, at least in principle, be generalised to genuinely
nonlinear potentials as well. (The double-well nature of $\Phi$ describes the co-existence of
different stable states and thus the possibility of interfaces between those states.)

The potential~\eqref{TDMF:Potential} is normalised to have unit sound speed, $c_0=1$. As
illustrated in Figure~\ref{fig:types}, there is a critical velocity $c_2>0$ such that for all
with $\speed$ with $c_2<\abs{\speed}<1$ there is a unique solution $\kappa>0$ to
\begin{align}
  \label{TDMF:DispRel}
  \speed^2 k^2=\Om^2\at{k},\quad
  \Om\at{k}=2\sin\at{k/2}.
\end{align}
From now on we solely consider waves with $\abs{\speed}>c_2$ because then the tails are
periodic with a unique wave number $\kappa$ as chosen above. This means, for any travelling wave with a single interface we have
\begin{align}
  \label{TDMF:Convergence}
  \pair{\calR}{\calV}\at\phase\xrightarrow{\phase\to\pm\infty}
  \pair{\calR_{\pm\infty}}{\calV_{\pm\infty}}\at\phase,
\end{align}
where both $\pair{\calR_{+\infty}}{\calV_{+\infty}}$ and
$\pair{\calR_{-\infty}}{\calV_{-\infty}}$ are periodic travelling waves (possibly constant) with phase speed
$\speed$ and group speed $\grspeed$. To compute the thermodynamic fields explicitly, it is
necessary that the asymptotic microscopic strains are confined to the harmonic wells. We thus
require that both $\calR_{-\infty}$ and $\calR_{+\infty}$ have a definite sign. By symmetry
we can assume that $\calR_{\pm\infty}\at\phase\gtrless0$, and by shift invariance we can also
assume that $\calR\at{0}=0$. Thus, the interface moves along $j=\speed{t}$ and
$\xi=\speed\tau$ in the microscopic and macroscopic space-time coordinates, respectively.
Notice, however, that we have not fixed the sign of $\speed$, so the wave travels from
negative strain to positive strain for $\speed>0$, and the other way around for $\speed<0$.

\myparagraph{Macroscopic constraints for phase transition waves}

Under the assumption that travelling waves with a single interface 
as described above exist,
all thermodynamic fields
are constant on the left and on the right of the interface and are completely determined by
the periodic tail oscillations in~\eqref{TDMF:Convergence}. The macroscopic conservation laws
therefore reduce to jump conditions via $\partial_\tau\rightsquigarrow-\speed\jump{\,}$ and
$\partial_\xi\rightsquigarrow\jump{\,}$,
so the PDEs~\eqref{Eqn:Macro.Cons.Laws} transform into
\begin{align}
  \label{TDMF:JumpConditions}
  \speed\jump{R}=-\jump{V}
  ,\quad%
  \speed\jump{V}=\jump{P}
  ,\quad%
  \speed\jump{E_\non+E_\osc}=\jump{PV+Q}.
\end{align}
Note that the asymptotic jump and mean value of any
thermodynamic field $X$ are given by
\begin{align*}
  \jump{X}=X_{+\infty}-X_{-\infty},\quad
  \mean{X}=\tfrac{1}{2}\at{X_{+\infty}+X_{-\infty}},\quad
X_{\pm\infty}=\lim_{\xi\to\pm\infty}{X}\pair{\tau}{\xi}.
\end{align*}
We now express the  asymptotic values of all thermodynamic fields $X$ in terms of $R_{\pm\infty}$,
$V_{\pm\infty}$, $A_{\pm\infty}$, and the speeds $\speed$ and $\grspeed$. In this way, we
recover well-known jump conditions and kinetic relations for phase transition
waves~\cite{Truskinovsky:82a,Truskinovsky:93a}. 
The strategy of computing thermodynamic quantities as 
averages of atomic observables is well established, see for
instance~\cite{Truskinovsky:05a,Slepyan:05a}. However, it is usually not based on
Young measures and hyperbolic scaling limits. It further seems that the concepts of
oscillatory and non-oscillatory energy have not been used before in the context of phase transition waves.
\par%
Due to \eqref{TDMF:Potential} and the sign choice for $\calR_{\pm\infty}$, we have $P_{\pm\infty}=-\Phi^\prime\at{R_{\pm\infty}}=-{R}_{\pm\infty}\pm1$
and hence $\jump{P}=2-\jump{R}$. The jump conditions for
mass~\eqref{TDMF:JumpConditions}$_1$ and momentum~\eqref{TDMF:JumpConditions}$_2$ thus imply
\begin{align}
  \label{TDMF:MassMom}
  \jump{R}=\frac{2}{1-\speed^2}, \qquad
  \jump{V}=-\frac{2\speed}{1-\speed^2},
\end{align}
and therefore \begin{align*}
  \jump{E_\non} =
  2\frac{\speed^2\mean{R}-\speed\mean{V}}{1-\speed^2}
  ,\qquad%
  \jump{PV}=2\frac{\speed\mean{R}-\speed^2\mean{V}}{1-\speed^2}.
\end{align*}
Using this and the formulae for $E_\osc$ and $Q$ from~\eqref{TDTW:OscDensAndFlux}, we then find
\begin{align}
  \label{TDMF:Energy.NonOsc}
  -\speed\jump{E_\non}+\jump{PV}=2\speed\mean{R},\quad
  -\speed\jump{E_\osc}+\jump{Q}=\at{\grspeed-\speed}\jump{E_\osc},
\end{align}
which is the analogue to~\eqref{eq:production.e1}. Consequently, the jump condition for the
total energy~\eqref{TDMF:JumpConditions}$_3$ enforces that the productions for oscillatory
and non-oscillatory energy cancel via
\begin{align}
  \label{TDMF:OscEnergyProd}
  \Xi=-2\speed\mean{R}=\at{\grspeed-\speed}\jump{E_\osc}.
\end{align}
This formula is important as it reveals that for phase transition waves there is no
production of total energy but instead a steady transfer between the oscillatory and the
non-oscillatory contributions of the energy. This transfer has power $-2\mean{R}$ and drives
the wave. More precisely, the \emph{configurational force} $\Upsilon$ satisfies
\begin{align}
  \notag
  \speed{\Upsilon}=\Xi,\qquad
  \Upsilon:=\jump{\Phi\at{R}}-\mean{\Phi^\prime\at{R}}\jump{R}.
\end{align}
This is the~\emph{kinetic relation} and follows from~\eqref{TDMF:MassMom}
and~\eqref{TDMF:Energy.NonOsc} thanks to $\jump{E_\non}=\jump{V}\mean{V}+\jump{\Phi\at{R}}$,
$\jump{PV}=\mean{P}\jump{V}+\jump{P}\mean{V}$ and $P=-\Phi^\prime\at{R}$. The production of
oscillatory energy $\Xi$ is the process commonly called dissipation. Recall, however, that the total energy $E=E_\osc+E_\non$ is not dissipated but conserved according to the energy laws~\eqref{Eqn:Macro.Cons.Laws}$_3$ and~\eqref{TDMF:JumpConditions}$_3$. 

We finally notice that time reversal changes the sign of $\speed$, $\grspeed$, $\Xi$,
$V_{\pm\infty}$, $Q_{\pm\infty}$, but does not affect $R_{\pm\infty}$, $A_{\pm\infty}$,
$E_{\osc,\,\pm\infty}$, $E_{\non,\,\pm\infty}$, $P_{\pm\infty}$, or $\Upsilon$.  We also
observe that all thermodynamic fields are completely determined by
\begin{align}
\label{TDMF:Prms}%
  \speed,\quad
  A_{-\infty}
  ,\quad
  A_{+\infty},\quad
  \mean{V} .
\end{align}
In fact, from~\eqref{TDMF:Prms} we compute $ \kappa$ and $\grspeed$ by~\eqref{TDMF:DispRel}
and set $\jump{E_\osc}=\tfrac{1}{2}\jump{A^2}$. Afterwards we solve~\eqref{TDMF:MassMom}$_1$
and~\eqref{TDMF:OscEnergyProd} for $R_{-\infty}$ and $R_{-\infty}$, which then allow us to
compute $V_{-\infty}$ and $V_{+\infty}$ from~\eqref{TDMF:MassMom}$_2$.

The jump conditions derived in this section constitute \emph{macroscopic constraints} which are \emph{necessary}
for the existence of a phase transition wave with speed $\abs{\speed}\in\oointerval{c_2}{1}$.
However, it was proven in~\cite{Schwetlick:08a} that these conditions are also
\emph{sufficient}, at least for near sonic speeds with $0<1-\abs{\speed}\ll1$. In conclusion, there exists a four-parameter family of candidates for phase transition waves, that means of heteroclinic travelling waves. It is now very natural to ask which of them are physically reasonable, and so selection criteria come into
play.

\myparagraph{The macroscopic aspects of Sommerfeld's radiation conditions}

We first consider~\eqref{eq:SOM1}. It is reasonable to require that the interface is a source
rather than a sink of oscillatory energy. The production $\Xi$ therefore has to be non-negative, which means
\begin{align}
  \label{TDMF:Radiation.Condition1}
  \Xi = (\grspeed-\speed)\jump{E_\osc}\geq0.
\end{align}
This inequality is \emph{equivalent} to
\begin{align}
  \label{TDMF:EntropyCondition}
  \speed\Upsilon\geq0,
\end{align}
which is the usual \emph{entropy condition} for phase transition waves (see for
example~\cite{Truskinovsky:05a}).  For all waves considered here,
\eqref{TDMF:Radiation.Condition1} implies $E_{\osc,\,+\infty}<E_{\osc,\,-\infty}$
for waves moving to the right and $E_{\osc,\,+\infty}>E_{\osc,\,-\infty}$ for left-moving
waves.  In other words, in all cases we have $\Xi>0$ if and only if the oscillations have smaller amplitude
in front of the interface than behind the interface. \eqref{eq:SOM1} select these solutions
but rejects waves that travel from regions of low oscillations into regions of high
oscillations. Note, however, that oscillations in front of the interface are not ruled out
since it is only required that the wave propagates in direction of decreasing oscillations.
This implies that there is still a four-parameter family of phase transition waves which
satisfy~\eqref{eq:SOM1}.

\begin{figure}[t!]
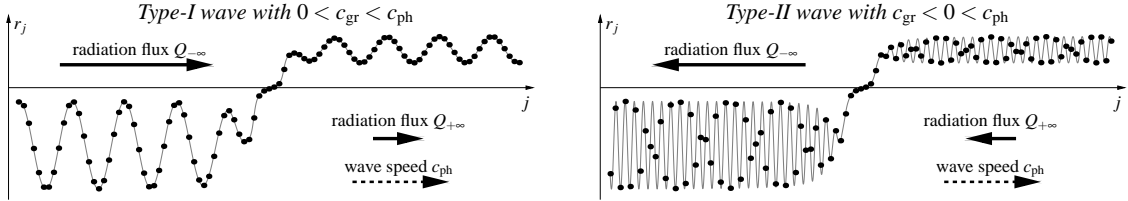

  \centering{\begin{tabular}{c}%
    \includegraphics[width=0.45\textwidth]%
    {\figfile{front_type_I}}%
    \hspace{.05\textwidth}%
    \includegraphics[width=0.45\textwidth]%
    {\figfile{front_type_II}}%
  \end{tabular}}%
\caption{%
Phase transition waves are driven by a constant transfer between the oscillatory and the
non-oscillatory energy and the radiation fluxes on both sides of the interface have the same sign. The cartoons illustrate the
\emph{source} solutions for type-I and type-II waves, which represent different
order relation for the group velocity $\grspeed$ and the phase velocity $\speed$. If time is
reversed, the interface becomes a sink of oscillatory energy and the radiation fluxes on both
sides change their sign. } \label{fig:cartoon_2}
\end{figure}

\begin{figure}[t!]
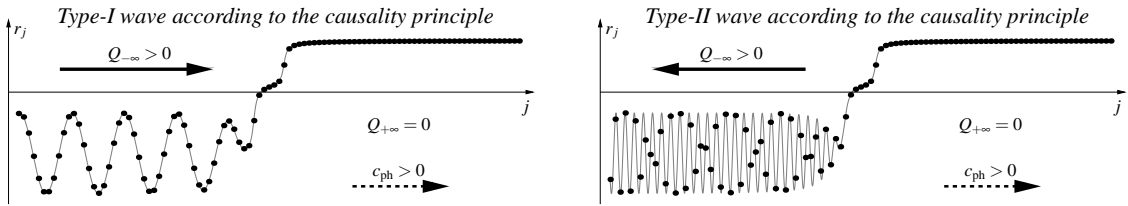

  \centering{\begin{tabular}{c}%
    \includegraphics[width=0.45\textwidth]%
    {\figfile{front_cold_I}}%
    \hspace{.05\textwidth}%
    \includegraphics[width=0.45\textwidth]%
    {\figfile{front_cold_II}}%
\end{tabular}}%
\caption{The causality principle selects phase  transition waves that propagate
into a region without oscillations. These waves are, up to Galilean invariance, uniquely determined by their speed $\speed$. }
  \label{fig:cartoon_3}
\end{figure}

Sommerfeld's second formulation~\eqref{eq:SOM2}, which stipulates that ``energy is carried
away from the interface'', translates directly into a condition on the radiation flux. It
requires, on both sides of the interface, that $Q$ points away from the interface. This
condition is very restrictive for phase transition waves with periodic tails because both
$Q_{-\infty}$ and $Q_{+\infty}$ have the same sign as the group velocity $\grspeed$,
see~\eqref{TDTW:OscDensAndFlux}. Thus~\eqref{eq:SOM2} can only be satisfied if there are no
oscillations on one side of the interface. However, which side of the interface is chosen by ~\eqref{eq:SOM2} depends on the sign of
$\grspeed$, and therefore we distinguish between two types, see Figures~\ref{fig:types}
and~\ref{fig:cartoon_2}. Type-I waves have $c_1<\abs{\speed}<1$, where
$c_1:=2\Omega\at{\pi/2}/\pi$, and this implies $\sgn\grspeed=\sgn\speed$ and $\abs{\grspeed}<\abs{\speed}$. Type-II waves
correspond to $c_2<\abs{\speed}<c_1$, which means $\sgn{\grspeed}\neq\sgn{\speed}$.
\par%
For type-I waves, the radiation fluxes behind and in front of the interface point towards and
away from the interface, respectively. This is illustrated in Figure~\ref{fig:cartoon_2}, and
holds regardless whether~\eqref{TDMF:Radiation.Condition1} is satisfied or not. Energy is
therefore always radiated towards the interface, and the second formulation of the radiation
condition can only be satisfied if there are no oscillations behind the interface. Those
waves, however, are usually regarded as unphysical as they
violate \eqref{TDMF:Radiation.Condition1} and~\eqref{TDMF:EntropyCondition}. The only solution candidates that would be accepted by
both formulations have no oscillations, neither in front nor behind the interface; however,
such single transition waves do not exist for the potential~\eqref{TDMF:Potential}.
\par%
The discussion is different for type-II waves. There is still radiation into the interface
but now the radiation flux impinges from ahead of the interface. Therefore, both~\eqref{eq:SOM1}
and~\eqref{eq:SOM2} are simultaneously  satisfied by type-II waves that propagate into a region
without oscillations, that is, $E_{\osc,\,+\infty}=0$ for right-moving waves and 
$E_{\osc,\,-\infty}=0$ for left-moving waves.  In particular,  there exists a
 three-parameter family of type-II waves that satisfy both~\eqref{eq:SOM1} and~\eqref{eq:SOM2}.  

In summary, for type-I waves, and hence for near sonic waves, \eqref{eq:SOM1} and~\eqref{eq:SOM2} contradict each other and would, if applied together, reject any bounded phase transition wave with a single interface. For type-II waves, however, 
\eqref{eq:SOM2} implies \eqref{eq:SOM1} and allows for a three-parameter family of waves. This is in stark contrast to the case of a standing source discussed in \sect\ref{sec:Stat-Somm-probl}, where both criteria are equivalent. 

\myparagraph{Microscopic selection criteria}

Besides macroscopic criteria as described above, there also exist microscopic selection
principles for phase transition waves. These are far more restrictive and select a
two-parameter family of phase transition waves, as shown below. For the sake of comparison we now summarise
the main arguments leading to microscopic selection criteria for phase transition waves in
bi-quadratic NN chains and refer to \cite{Truskinovsky:05a,Cherkaev:05a} for more details. The
key idea is that under the condition $\sgn{\calR\at{\phase}}=\sgn{\phase}$, each phase
transition wave is determined by the affine advance-delay-differential equation
\begin{align*}
  \speed^2\partial_\phase^2\calR
  =\triangle\calR-\triangle\sgn{\phase}.
\end{align*}
This equation can be regarded as the analogue to the inhomogeneous Helmholtz
equation~\eqref{Eqn:NN.Inhom.Helmholtz}, and solutions can be represented by
$\calR\at\phase=\mathrm{Re}\at{\calS\at\phase}$ with
\begin{align}
  \label{PTW:ContourIntegral}
  \calS\at\phase=-\frac{\iu}{\pi}\int\limits_\Gamma
  \frac{\Omega\at{k}^2\mhexp{+\iu{k}\phase}}{k\Omega^2\at{k}
    -\speed^2{k}^3}\, \dint{k},
\end{align}
where $\Gamma$ is an appropriately chosen contour in the complex plane. The microscopic
selection criterion is based on the \emph{causality principle}~\cite{Slepyan:01a,Slepyan:02a} and requires
that all oscillatory modes in front and behind the interface satisfy $\grspeed > \speed$ and  $\grspeed < \speed$, respectively. The contour
$\Gamma$ is therefore chosen  as
the dented real axis that passes the origin $k=0$ from below but the other real-valued poles
of the integrand in~\eqref{PTW:ContourIntegral} from above.  Jordan's 
lemma from complex-valued calculus then provides the following expressions for the thermodynamic 
fields for a right moving wave
\begin{align}
\label{CausPrinciple1}
  R_{\pm\infty}=\pm\frac{\speed}{1-\speed^2}+\frac{\speed}{\grspeed-\speed},
  \qquad
  A_{-\infty}=\frac{2\speed}{\grspeed-\speed},\quad
  A_{+\infty}=0,
\end{align}
with $\jump{V}$ as in \eqref{TDMF:MassMom}, and therefore
\begin{align}
\label{CausPrinciple2}
  \jump{E_\osc}=-\frac{2\speed^2}{\at{\grspeed-\speed}^2}<0,\qquad
  \Xi=\frac{2\speed^2}{\speed-\grspeed } >0.
\end{align}
In particular, there exists a two-dimensional family of phase transition waves that is
parameterised by the speed $\speed$ and the trivial parameter $\mean{V}$.  All these causality waves
have no oscillations ahead of the interface and~\eqref{eq:SOM1}
is always satisfied. The validity of~\eqref{eq:SOM2}, however, depends on 
$\sgn\grspeed=\sgn{Q}$, i.e., on whether the wave is of type-I or type-II.  This is illustrated by the two cartoons from Figure~\ref{fig:cartoon_3}.  We therefore find again that~\eqref{eq:SOM2} has different implications for the two wave types. This is not surprising since  a condition on the sign of $\speed-\grspeed$ does not imply any constraint for $\sgn{\grspeed}=\sgn{Q}$, or~\emph{vice versa}. In other words, \eqref{eq:SOM2} does not imply the causality principle, or \emph{vice versa}. For a type-I wave the causality principle even implies that energy is radiated towards the interface from behind, see the left panel in Figure~\ref{fig:cartoon_3}.
\bigpar
We recall that  for standing sources as discussed in~\sect\ref{sec:Stat-Somm-probl}, the causality principle in fact implies~\eqref{eq:SOM2}. It is therefore tempting to adopt the criterion~\eqref{eq:SOM2} for moving
phase interfaces and reformulate this condition as ``the energy flux has to point away from the interface, with reference to an observer travelling with the interface''. Let us follow this argument for causality waves as shown in Figure~\ref{fig:cartoon_3}. Changing to the co-moving frame via $\tilde j = j - \speed t$ and $\tilde{\xi}=\xi-\speed\tau$, the
partial energy balances~\eqref{eq:production.e1} transform into
\begin{align*}
  \partial_\tau{E_\osc}+\partial_{\tilde\xi}\at{Q-\speed{E}_\osc}=+\Xi,\qquad
\partial_\tau{E_\non}+\partial_{\tilde\xi}\at{PV-\speed{E}_\non}=-\Xi.
\end{align*}
In view of $Q=\grspeed{E_\osc}$ we now conclude that
the causality principle implies that the \emph{relative radiation flux} $\tilde{Q}=Q-\speed{E_\osc}$ has the same sign as $\tilde\xi$, and is hence indeed pointing away from the interface on both sides. The flaw with this argument, however,
is that the
passage to the co-moving frame is not a Galilei transformation since both $j$ and $\xi$ are Lagrangian space coordinates (Galilei transformations
are given by $u_j\mapsto u_j+ct$). The obervation that $\tilde{Q}$ has a certain sign on both sides of the interface does therefore not imply that the interface radiates energy towards infinity.  The correct treatment is to work with the Lagrangian space coordinates $j$ and $\xi$ and to characterise the radiative parts of the energy flux in terms of $Q$. This flux can, as shown in Figure~\ref{fig:cartoon_3}, point towards the interface from 
the tail of the wave.

\myparagraph{Selection rules from Riemann problems}

A different approach to microscopic selection criteria is via initial values problems. The main idea
is to characterise the physically relevant solutions as 
the limit as
$t\to\infty$ of  reasonable initial data. For the NN chains with double well potential, a rigorous mathematical analysis of initial values problems is not yet available, but numerical simulations provide a lot of insight into the implied selection criteria. 

To illustrate this, we now present two numerical solutions of initial value problems for the NN chain with interaction potential~\eqref{TDMF:Potential}. Due to the Hamiltonian nature of~\eqref{Eqn:ForcedNN}, we perform the simulations with the Verlet scheme (see, e.g.,~\cite{HLW02} for details), which is a symplectic integrator and does not add numerical viscosity to the problem. In order to enforce the formation of macroscopic waves, we start with Riemann initial data, that means we set

\begin{align*}
r_j\at{0}=r_{\pm\infty},\qquad
v_j\at{0}=v_{\pm\infty}\qquad \text{for}\qquad j\gtrless0,
\end{align*}

with prescribed asymptotic data $r_{\pm\infty}$, $v_{\pm\infty}$. Moreover, we choose a `number of particles' $N\gg1$ and introduce the scaling parameter $\eps=1/N$. We then integrate the lattice equation on the domain $\abs{j}\leq{N/2}$ over the time interval $0\leq\bar{t}_\fin{N}$ with $0<\bar{t}_\fin<1/2$. To this end we close the discrete equations~\eqref{eq:DCL.MassMom} by 
imposing Dirichlet boundary data for $r$ and $v$. Notice that the choice of the computational domain and time is in accordance with both the hyperbolic scaling and $c_0=1$, this means no macroscopic wave can hit the boundary of the computational domain.
\par%
Figures~\ref{fig:rp_1} and~\ref{fig:rp_2} depict the numerical results at time $\bar{t}_\fin$ for two Riemann problems with $\sgn\at{r_{-\infty}}\neq\sgn\at{r_{+\infty}}$. The left panel contains
snapshots of the atomic distances against the scaled particle index $\bar{j}=\eps{j}$,
and the other panels show
the spatial profiles of macroscopic strain $R$ and heat flux $Q$, which are computed by mesoscopic averaging; see~\cite{DH08} for details.
\bigpar%
\begin{figure}[ht!]
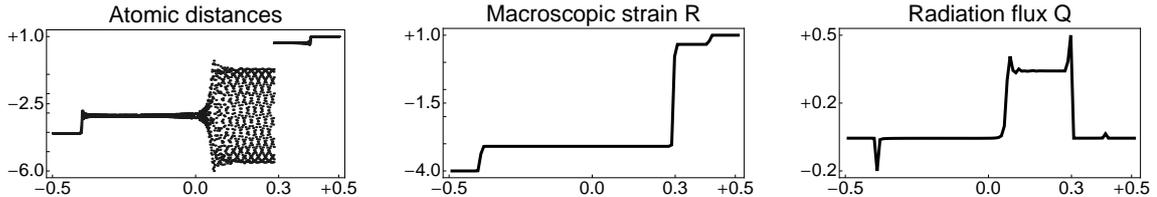

  \centering{
    \includegraphics[width=0.3\textwidth]%
    {\figfile{rp1_atom_dist}}\hspace{0.01\textwidth}%
     \hspace{0.025\textwidth}%
    \includegraphics[width=0.3\textwidth]%
    {\figfile{rp1_fld_dist}}\hspace{0.01\textwidth}%
     \hspace{0.025\textwidth}%
    \includegraphics[width=0.3\textwidth]%
    {\figfile{rp1_fld_heat_flux}}\hspace{0.01\textwidth}%
}%
\caption{Numerical solutions to the Riemann problem with
$r_{-\infty}=-4$, $r_{+\infty}=+1$,  $v_{-\infty}=+1$, $v_{+\infty}=-1$ and $N=4000$, $\bar{t}_\fin=0.4$.
The microscopic dynamics selects a type-I wave satisfying the causality principle.}
 \label{fig:rp_1}
\end{figure}
\begin{figure}[ht!]
  \centering{%
    \includegraphics[width=0.3\textwidth]%
    {\figfile{rp2_atom_dist}}\hspace{0.01\textwidth}%
    \hspace{0.025\textwidth}%
    \includegraphics[width=0.3\textwidth]%
    {\figfile{rp2_fld_dist}}\hspace{0.01\textwidth}%
     \hspace{0.025\textwidth}%
    \includegraphics[width=0.3\textwidth]%
    {\figfile{rp2_fld_heat_flux}}\hspace{0.01\textwidth}%
}%
\caption{Numerical solutions the Riemann problem with
$r_{-\infty}=-3/2$, $r_{+\infty}=+1$,  $v_{-\infty}=+1$, $v_{+\infty}=-1$ and $N=4000$, $\bar{t}_\fin=0.4$. The microscopic dynamics selects a type-II wave satisfying the causality principle.}
  \label{fig:rp_2}
\end{figure}
In both simulations we observe four elementary waves connecting five constant states. Starting from the left side we find:  Left initial state, first wave, first non-oscillatory intermediate state, second wave, oscillatory intermediate state, third wave, second non-oscillatory intermediate state, fourth wave, and the right initial state. The first and the fourth wave, which both connect two non-oscillatory states, are contact discontinuities and correspond to jumps in the linear wave equation. The second wave, which connects to the oscillatory intermediate state, is still a harmonic phenomenon since the oscillations are confined to one of the quadratic wells of $\Phi$. We also notice that the macroscopic strain does not change within the second wave.
\par%
The third wave, however, describes a phase transition as it connects negative strain to positive strain. A closer look to the oscillatory intermediate state reveals that it is generated by a single harmonic travelling wave which has a certain phase velocity $\speed$. The third wave propagates with the same speed $\speed$, and hence it corresponds in fact to a heteroclinic travelling wave in the lattice. We finally mention that for $\eps>0$ we still have harmonic fluctuations within the constant intermediate state and near the contact discontinuities, but these disappear in the thermodynamic limit $\eps\to0$. 
\par%
Figures~\ref{fig:rp_1} and~\ref{fig:rp_2} reveal that the radiation flux $Q$ within the oscillatory intermediate state has a different sign in both simulations. More precisely, we see  a type-I wave in  Figure~\ref{fig:rp_1} but a  type-II wave in Figure~\ref{fig:rp_2}. Recall that both waves are now `chosen' by the microscopic dynamics (more precisely, by the numerical scheme). In other words, the
macroscopic structure of the Riemann solution (the parameters of the waves and  
the intermediate states) determine --- implicitly,  but nevertheless uniquely --- a microscopic selection rule for phase transition waves. We are not able to 
derive this selection rule rigorously from the lattice dynamics although this should, in principle, be possible. However, 
computing the speed of the interface as well as thermodynamic fields
on both sides of the interface, we find that these meet very well the predictions from the causality principle described above, see~\eqref{CausPrinciple1} and~\eqref{CausPrinciple2}. It would be highly desirable to
understand this in greater detail and to find a reformulation of the causality principle that can also be applied to generic double well potentials, where Fourier methods can no longer be used. First steps into this direction are already done in this paper,
because our thermodynamic approach to radiation conditions can be generalised
to more complicate nonlinear problems. However, we emphasise that many questions remain open.

\myparagraph{Discussion}

The discussion in this section shows that the case of a moving interface is different from that of a standing source; for the latter,~\eqref{eq:SOM1} and~\eqref{eq:SOM2} are equivalent. For moving interfaces in phase transition waves we find that~\eqref{eq:SOM1}, which is 
formulated in terms of sources and
sinks, is equivalent to the entropy condition~\eqref{TDMF:EntropyCondition}.
The flux condition~\eqref{eq:SOM2}, however, implies~\eqref{eq:SOM1} for type-II waves but both rules contradict each other for type-I waves.
\par
For moving interfaces, one has therefore to
distinguish between arguments that rely on the energy transport in terms of fluxes, as~\eqref{eq:SOM2}, and
arguments based on energy productions, as~\eqref{eq:SOM1}. 
Since~\eqref{eq:SOM1} is equivalent to the entropy inequality, we propose to apply this selection criterion for moving inhomogeneities. We also propose that 
flux based criteria such as~\eqref{eq:SOM2} should \emph{not} be na\"\i vely applied to problems with moving inhomogeneities. 
\par
These suggestions are in line with the causality 
principle used in~\cite{Slepyan:01a,Slepyan:02a,Truskinovsky:05a} and numerical simulations of  
Riemann-type problems for NN chains.

\appendix
\section{Appendix: Macroscopic conservation laws for NN chains}
\label{app:1}
We establish the thermodynamic limit for the forced NN chain ~\eqref{Eqn:ForcedNN} provided that the forcing $\zeta$ satisfies~\eqref{App:Forcing}. Our goal is to show
that the macroscopic balance laws for mass, momentum and energy (see~\eqref{eq:NN.Limit}
and~\eqref{Eqn.EnergyBalanceWithForcing}) can be derived rigorously as follows:
\begin{enumerate}
\item
The hyperbolic scaling transforms each \emph{bounded} solution into a family of
\emph{oscillatory} functions which depend   on the macroscopic time $\tau$ and
macroscopic Lagrangian space   coordinate $\xi$.
\item
This family of functions is compact in the sense of Young   measures, and hence we can
extract convergent subsequences. Along such a subsequence, the limit measure encodes the
\emph{local distribution functions} of the oscillatory data and hence the \emph{local
mean values} of atomic observables. These local mean values provide the
\emph{thermodynamic fields} and are, by construction, \emph{non-oscillatory} functions in
$\tau$ and $\xi$.
\item
The dynamics of NN chains implies that the thermodynamic fields of each Young measure limit
satisfy the macroscopic conservation laws of mass, momentum and energy (in a distributional sense).
\end{enumerate}
We now collect the mathematical tools for each of these steps. We start with some basic facts
about Young measures and refer the reader
to~\cite{Ball:89a,Roubicek:97a,Valadier:94a,Taylor:97c} for more details.

Let $\Omega$ be a domain in $\Rset^k$ and $K$ be some convex and closed set in $\Rset^m$. A
Young measure $\mu\in{Y}\at{\Omega;K}$ is a $\Omega$-family of probability measures on $K$,
that means a measurable map $\mu\colon
y\in\Omega\to\mu\pair{y}{\dint{Q}}\in\fspace{Prob}\at{K}$. Notice that each function
$Q\colon\Omega\to{K}$ defines a trivial Young measure with
$\mu\pair{y}{\dint{Q}}=\delta_{Q\at{y}}\at{\dint{Q}}$, where $\delta_{Q\at{y}}\at{\dint{Q}}$
abbreviates the delta distribution in $Q\at{y}$.

\begin{theorem}[Fundamental Theorem on Young Measures]
\label{App:Theo1} %
Each family
\begin{align*}
\at{Q_\eps}_{0<\eps\leq{1}}\subset{}\fspace{L}^\infty\at{\Omega;K}
\end{align*}
is compact in the space of Young-measures   $\fspace{YM}\at{\Omega;K}$.  This means there
exists a sequence $(\eps_n)_{n \in \N}$ with $\eps_n\to0$ along with a limit measure
$\mu\in\fspace{YM}\at{\Omega;K}$ such that
  \begin{align}
    \label{App:Theo1.Eqn1}
    \psi\at{Q_{\eps_n}}\xrightarrow{n\to\infty}\average{\psi}%
    \quad\text{weakly$\star$ in
      $\fspace{L}^\infty\at{\Omega}$}
  \end{align}
  for all observables $\psi\in\fspace{C}\at{{K}}$, where
  \begin{align*}
    \average{\psi}\at{y}=
    \int\limits_{B}\psi\at{Q}\mu\pair{y}{\dint{Q}}
  \end{align*}
  gives the local mean value of $\psi$ in $y\in\Omega$.
\end{theorem}

\begin{proof}
  See, for instance,~\cite{Taylor:97c}, Proposition 11.3 in Section   13.11.
\end{proof}

The convergence~\eqref{App:Theo1.Eqn1} is equivalent to
\begin{align}
\label{App:YMLimitFormula}
  \int\limits_{\Omega}\average{\psi}\at{y}\phi\at{y}{\,\dint{y}}=\lim_{
    n\to\infty}\int\limits_{\Omega}\psi\at{ Q_{\eps_n}\at{y} }
  \phi\at{y}{\,\dint{y}}
\end{align}
for all test functions $\phi\in\fspace{C}_{\comp}^\infty\at{\Omega}$.  Moreover, the
subsequence converges strongly to some limit function $Q_0$ in
$\fspace{L}^\infty\at{\Omega;K}$ if and only if the limit measure is trivial,
\begin{math}
  \mu\pair{y}{\dint{Q}}=\delta_{Q_0\at{y}}\at{\dint{Q}}.
\end{math}

We now suppose that we are given a bounded solution to~\eqref{Eqn:ForcedNN}. As in
\sect\ref{sec:Macr-field-equat}, we regard the atomic distances $r_j=u_{j+1}-u_j$ and
velocities $v_j=\dot{u}_j$ as the basic variables, i.e., we consider
\begin{align}
  \label{App:Micro.Sol}
  Q_j\at{t}=\pair{r_j\at{t}}{v_j\at{t}},\quad
  j\in\Zset,\quad{t\geq0}.
\end{align}
For a given scaling parameter $0<\eps\leq{1}$ we introduce $\tau$ and $\xi$
by~\eqref{eq:NN.Scaling}, so the macroscopic Lagrangian space-time coordinate is given by
\begin{align*}
  \Omega=\{\pair{\tau}{\xi}\;:\;\tau>0,\;\xi\in\Rset\}.
\end{align*}
Moreover, we identify~\eqref{App:Micro.Sol} with piecewise constant functions on $\Omega$,
\begin{align}
  \label{App:Identification}
  Q_\eps\pair{\eps{t}}{\eps{j}+\eta}=Q_{j}\at{t}\quad
  \text{ for every }
  t\geq0,\;{j}\in{\Zset},\;\abs{\eta}\leq\tfrac{1}{2}.
\end{align}
By assumption, we have $\at{Q_\eps}_{0<\eps\leq{1}}\subset{\fspace{L}^\infty\at{\Omega;K}}$
for some ball $K\subset\Rset^2$, and Theorem \ref{App:Theo1} provides at least one
subsequence that converges to some limit measure $\mu\in\fspace{YM}\at{\Omega;K}$. Moreover,
for each atomistic observable $\psi$ we can compute the corresponding thermodynamic field via
\begin{align*}
  \baverage{\psi\pair{r}{v}}\pair{\tau}{\xi}=
  \int_{\Rset^2}\psi\pair{r}{v}\mu\triple{\tau}{\xi}{\,\dint{r}\dint{v}}.
\end{align*}
We are now able to state and prove the main result on the thermodynamic limit of forced NN chains. It is a direct consequence of the discrete conservation laws derived from~\eqref{Eqn:ForcedNN}, the notion of Young-measure convergence, and the properties of distributional derivatives.

\begin{theorem}[Macroscopic conservation laws for NN chains]
\label{theo:macro-cons}
The thermodynamic fields of each Young measure limit $\mu$ of~\eqref{Eqn:ForcedNN} satisfy
the conservation laws~\eqref{eq:NN.Limit} within the following domains in the sense of
distributions: The laws for mass~\eqref{eq:NN.Limit}$_1$ and
momentum~\eqref{eq:NN.Limit}$_2$ hold for $\Omega$. The conservation of
energy~\eqref{eq:NN.Limit}$_3$ holds for $\Omega$ if $\zeta\equiv0$, and otherwise
for   $\tilde{\Omega}:=\Omega\setminus\{\xi=0\}$.
\end{theorem}

\begin{proof}
Within this proof we write
$Q_\eps\pair{\tau}{\xi}=\bpair{r_\eps\pair{\tau}{\xi}}{v_\eps\pair{\tau}{\xi}}$. The equation
of motion~\eqref{Eqn:ForcedNN}, combined with the scaling rules~\eqref{eq:NN.Scaling}
and~\eqref{App:Identification}, implies the following discrete conservation laws
  \begin{align}
    \label{App:DCL1}
    \partial_\tau{r}_\eps-\nabla_{+\eps}{v}_\eps&=0
    ,\\%
    \label{App:DCL2}
    \partial_\tau{v}_\eps-\nabla_{-\eps}\Phi^\prime\at{{r}_\eps}
    &=\zeta\at{\tau/\eps}\chi_{\eps}\at{\xi}
    ,\\%
    \label{App:DCL3}
    \partial_\tau\bat{\tfrac{1}{2}v_\eps^2
      +\Phi\at{r_\eps}-\eps\nabla_{-\eps}\Phi\at{r_\eps}}
    -\nabla_{-\eps}\bat{v_\eps\Phi^\prime\at{r_\eps}}
    &=\zeta\at{\tau/\eps}v_0\at{\tau/\eps}
    \chi_{\eps}\at{\xi}
  \end{align}
for all $\tau>0$ and almost all $\xi\in\Rset$, where the discrete differential operators
$\nabla_\pm\eps$ and the scaled cut off   function $\chi_\eps$ are given by
  \begin{align*}
    \at{\nabla_{\pm\eps}{f}}\pair{\tau}{\xi}=%
    \frac{\pm{f\pair{\tau}{\xi\pm\eps}\mp{f}\pair{\tau}{\xi}}}{\eps},
    \qquad\chi_\eps\at\xi=\frac{\chi_{\{\abs{\xi}<\eps/2\}}\at\xi}{\eps}.
  \end{align*}
We now multiply~\eqref{App:DCL1} with a test function
$\phi\in{\fspace{C}^\infty_\comp}\at{\Omega}$ and integrate with respect to both $\tau$ and
$\xi$. Using integration by parts and expansions with respect to $\eps$ we then find
\begin{align*}
    \int_\Omega
    r_\eps\partial_\tau\phi-v_\eps\partial_\xi\phi\,\dint\tau\dint\xi=O\at{\eps},
  \end{align*}
so the limit $\eps\to0$ provides~\eqref{eq:NN.Limit}$_1$ in the  sense of distributions, see \eqref{App:YMLimitFormula}.
Similarly, and using that~\eqref{App:Forcing} implies
  \begin{align*}
    \int_\Omega\zeta\at{\tau/\eps}\chi_{\eps}\at{\xi}
    \phi\pair{\tau}{\xi}\,\dint\tau\dint\xi=
    \eps\int_0^\infty\zeta\at{t}\phi\pair{\eps{t}}{0}\,\dint{t}+{O}\at\eps
    ={O}\at\eps,
  \end{align*}
  we derive~\eqref{eq:NN.Limit}$_2$ from~\eqref{App:DCL2}. Finally, the assertions about the
energy conservation follow from~\eqref{App:DCL3}, where for $\zeta\neq0$ we assume that all
test functions $\phi$ are compactly supported in $\tilde{\Om}$.
\end{proof}
Since the energy is conserved in $\tilde{\Omega}$ we can balance the energy in the whole domain
$\Omega$ via~\eqref{Eqn.EnergyBalanceWithForcing}.
%
\section*{Acknowledgements} MH was supported by the EPSRC Science and Innovation award to the
Oxford Centre for Nonlinear PDE (EP/E035027/1). JZ gratefully acknowledges funding from the
Royal Society (TG100352) and EPSRC (EP/H05023X/1, EP/F03685X/1).
%
%
%
%
\def\cprime{$'$} \def\cprime{$'$} \def\cprime{$'$}
  \def\polhk#1{\setbox0=\hbox{#1}{\ooalign{\hidewidth
  \lower1.5ex\hbox{`}\hidewidth\crcr\unhbox0}}} \def\cprime{$'$}
  \def\cprime{$'$}
\providecommand{\bysame}{\leavevmode\hbox to3em{\hrulefill}\thinspace}
\providecommand{\MR}{\relax\ifhmode\unskip\space\fi MR }
\providecommand{\MRhref}[2]{%
  \href{http://www.ams.org/mathscinet-getitem?mr=#1}{#2}
}
\providecommand{\href}[2]{#2}

\end{document}